\documentclass[11pt]{article}
\usepackage{epsf}
\usepackage{amsmath}
\usepackage{epsfig}
\usepackage{times}
\usepackage{amssymb}
\usepackage{amsthm}
\usepackage{setspace}
\usepackage{cite}

\usepackage{algorithmic}  % This package provides an algorithmic environment fo describing algorithms.
\usepackage{algorithm}

\usepackage{shadow}
\usepackage{fancybox}
\usepackage{fancyhdr}

\def\w{{\bf w}}

\def\y{{\bf y}}

\def\x{{\bf x}}

\def\x{{\mathbf x}}

\def\w{{\bf w}}

\def\x{{\bf x}}
\def\y{{\bf y}}
\def\z{{\bf z}}
\def\q{{\bf q}}
\def\a{{\bf a}}

\def\h{{\bf h}}

\def\be{\begin{equation}}
\def\ee{\end{equation}}
\def\ba{\left[\begin{array}}
\def\ea{\end{array}\right]}

\def\w{{\bf w}}

\def\x{{\bf x}}
\def\y{{\bf y}}
\def\z{{\bf z}}
\def\q{{\bf q}}
\def\a{{\bf a}}

\def\1{{\bf 1}}

\def\W{{\bf W}}
\def\G{{\bf G}}

\def\g{{\bf g}}
\def\0{{\bf 0}}

\def\w{{\bf w}}

\def\y{{\bf y}}

\def\x{{\bf x}}

\def\x{{\mathbf x}}

\def\w{{\bf w}}

\def\x{{\bf x}}
\def\y{{\bf y}}
\def\z{{\bf z}}
\def\q{{\bf q}}
\def\a{{\bf a}}

\def\h{{\bf h}}

\def\be{\begin{equation}}
\def\ee{\end{equation}}
\def\ba{\left[\begin{array}}
\def\ea{\end{array}\right]}

\def\w{{\bf w}}

\def\x{{\bf x}}
\def\y{{\bf y}}
\def\z{{\bf z}}
\def\q{{\bf q}}
\def\a{{\bf a}}

\def\W{{\bf W}}
\def\G{{\bf G}}

\def\g{{\bf g}}
\def\0{{\bf 0}}

\def\1{{\bf 1}}

\def\W{{\bf W}}
\def\G{{\bf G}}

\def\X{{\bf X}}

\def\g{{\bf g}}
\def\0{{\bf 0}}

\def\htheta{\hat{\theta}}

\def\H{{\bf H}}

\def\gammainc{\gamma_{inc}}

\def\cweak{c_{w}}

\def\betaweak{\beta_{w}}
\def\thetaweak{\theta_{w}}
\def\hthetaweak{\hat{\theta}_{w}}

\def\htheta{\hat{\theta}}

\def\cweak{c_{w}}

\newtheorem{theorem}{Theorem}

\newtheorem{lemma}{Lemma}

\begin{document}

\begin{singlespace}

\title {Optimality of $\ell_2/\ell_1$-optimization block-length dependent thresholds
%\footnote{ This work was supported in
%part.}
}
\author{
\textsc{Mihailo Stojnic}
\\
\\
{School of Industrial Engineering}\\
{Purdue University, West Lafayette, IN 47907} \\
{e-mail: {\tt mstojnic@purdue.edu}} }
\date{}
\maketitle

\centerline{{\bf Abstract}} \vspace*{0.1in}

The recent work of \cite{CRT,DonohoPol} rigorously proved (in a large dimensional and
statistical context) that if the number of equations (measurements
in the compressed sensing terminology) in the system is proportional
to the length of the unknown vector then there is a sparsity (number
of non-zero elements of the unknown vector) also proportional to the
length of the unknown vector such that $\ell_1$-optimization algorithm
succeeds in solving the system. In more recent papers \cite{StojnicCSetamBlock09,StojnicICASSP09block,StojnicJSTSP09} we considered under-determined systems with the so-called \textbf{block}-sparse solutions. In a large dimensional and
statistical context in \cite{StojnicCSetamBlock09} we determined
lower bounds on the values of allowable sparsity for any given
number (proportional to the length of the unknown vector) of
equations such that an $\ell_2/\ell_1$-optimization algorithm succeeds in solving the system. These lower bounds happened to be in a solid
numerical agreement with what one can observe through numerical experiments. Here we derive the  corresponding upper bounds. Moreover, the upper bounds that we obtain in this paper match the lower bounds from \cite{StojnicCSetamBlock09} and ultimately make them optimal.

\vspace*{0.25in} \noindent {\bf Index Terms: Linear systems of equations;
$\ell_2/\ell_1$-optimization; compressed sensing} .

\end{singlespace}

%%%%%%%%%%%%%%%%%%%%%%%%%%%%%%%%%%%%%%%%%%%%%%%%%%%%%%%%%%%%%%%%%
\section{Introduction}
\label{sec:back}
%%%%%%%%%%%%%%%%%%%%%%%%%%%%%%%%%%%%%%%%%%%%%%%%%%%%%%%%%%%%%%%%%

In last several years the area of compressed sensing has been the
subject of extensive research. Finding the sparsest
solution of an under-determined system of linear equations is one of the focal points of the entire area. Phenomenal results of \cite{CRT,DonohoPol} rigorously proved for the first time that in certain scenarios one can solve an under-determined system of linear equations by solving a linear program in polynomial time. These breakthrough
results then as expected generated enormous amount of research with possible
applications ranging from high-dimensional geometry, image
reconstruction, single-pixel camera design, decoding of linear
codes, channel estimation in wireless communications, to machine
learning, data-streaming algorithms, DNA micro-arrays,
magneto-encephalography etc. (more on the compressed sensing
problems, their importance, and wide spectrum of different
applications can be found in excellent references
\cite{DDTLSKB,JRimaging,BCDH08,CRchannel,VPH,WM08,Olgica,RFPrank}).

In this paper we will be interested in solving the following under-determined system of linear equations \vspace{-0.12in}
\begin{equation}
A\x=\y \label{eq:system}
\end{equation}\vspace{-0.27in}

\noindent where $A$ is an $M\times N$ ($M<N$) measurement matrix and $\y$ is
an $M\times 1$ measurement vector. Moreover, we will be interested in finding the $K$-sparse solution $\x$.
Under $K$-sparse we will in the rest of the paper consider vectors that have at most $K$ nonzero
components. Also throughout the rest of the paper we will often refer to the $K$-sparse solution of (\ref{eq:system}) simply as the solution of (\ref{eq:system}).
Further, we will consider ideally sparse signals; more on the so-called
approximately sparse signals can be found in e.g. \cite{DeVorenoise,
XHapp}. We will also assume the
so-called \emph{linear} regime, i.e. we will assume that $K=\beta N$
and that the number of the measurements is $M=\alpha N$ where
$\alpha$ and $\beta$ are absolute constants independent of $N$.

A particularly successful approach to solving (\ref{eq:system}) assumes
solving the following $\ell_1$-optimization problem\vspace{-0.08in}
\begin{eqnarray}
\mbox{min} & & \|\x\|_{1}\nonumber \\
\mbox{subject to} & & A\x=\y. \label{eq:l1}
\end{eqnarray}\vspace{-0.25in}

\noindent While the literature on using (\ref{eq:l1}) to solve
(\ref{eq:system}) is rapidly growing below we restrict our attention
to two, in our view, the most influential recent works \cite{CRT,DonohoPol}.
Quite remarkably, for certain statistical matrices $A$ in \cite{CRT,DonohoPol} the authors were able to show that if
$\alpha$ and $N$ are given then
any unknown vector $\x$ with no more than $K=\beta N$ (where $\beta$
is an absolute constant dependent on $\alpha$ and explicitly
calculated in \cite{CRT,DonohoPol}) non-zero elements can be recovered by
solving (\ref{eq:l1}). As expected, this assumes that $\y$ was in
fact generated by that $\x$ and given to us. (More on practically very important scenario when the
available measurements are noisy versions of $\y$ can be found in seminal works \cite{CRT,W} as well as in recent developments e.g. \cite{StojnicGenLasso10,StojnicGenSocp10,StojnicPrDepSocp10}.)

%%%%%%%%%%%%%%%%%%%%%%%%%%%%%%%%%%%%%%%%%%%%%%%%%%%%%%%%%%%%%%%%%
\section{Block-sparse signals and $\ell_2/\ell_1$-algorithm}
\label{sec:alg}
%%%%%%%%%%%%%%%%%%%%%%%%%%%%%%%%%%%%%%%%%%%%%%%%%%%%%%%%%%%%%%%%%

What we described in the previous section assumes solving an under-determined system of linear equations with a standard restriction that the solution vector is sparse. Sometimes one may however encounter
applications when the unknown $\x$ in addition to being sparse has a
certain structure as well. The so-called block-sparse vectors are such a type of vectors and will be the main subject of this paper. These vectors and their potential applications and recovery algorithms were
investigated to a great detail in a series of recent references (see e.g. \cite{EldBol09,EKB09,SPH,FHicassp,EMsub,BCDH08,StojnicICASSP09block,StojnicJSTSP09,GaZhMa09,CeInHeBa09}). A related problem
of recovering jointly sparse vectors and its applications were also
considered to a great detail in e.g. \cite{ZeGoAd09,ZeWaSeGoAd08,TGS05,BWDSB05,CH06,CREKD,MEldar,Temlyakov04,VPH,BerFri09,EldRau09,BluDav09,NegWai09} and many
references therein. While various other structures as well as their applications gained significant interest over last few years we here refrain from describing them into fine details and instead refer to nice work of e.g.
\cite{KDXH09,XKAH09,KXAH09,RFPrank}. Since we will be interested in characterizing mathematical properties of solving linear systems that are similar to many of those mentioned above we just state here in brief that from a mathematical point of view in all these cases one attempts to improve the
recoverability potential of the standard algorithms (which are typically similar to the one described in the
previous section) by incorporating the knowledge of the unknown vector
structure.

To get things started we first introduce the block-sparse vectors.
The subsequent exposition will also be somewhat less cumbersome if we assume that
integers $N$ and $d$ are chosen such that $n=\frac{N}{d}$ is an
integer and it represents the total number of blocks that $\x$
consists of. Clearly $d$ is the length of each block. Furthermore,
we will assume that $m=\frac{M}{d}$ is an integer as well and that
$\X_i=\x_{(i-1)d+1:id}, 1\leq i\leq n$ are the $n$ blocks of $\x$
(see Figure \ref{fig:blspmodel}).
\begin{figure}[htb]
%%%%%\begin{minipage}[b]{1.0\linewidth}
\centering
\centerline{\epsfig{figure=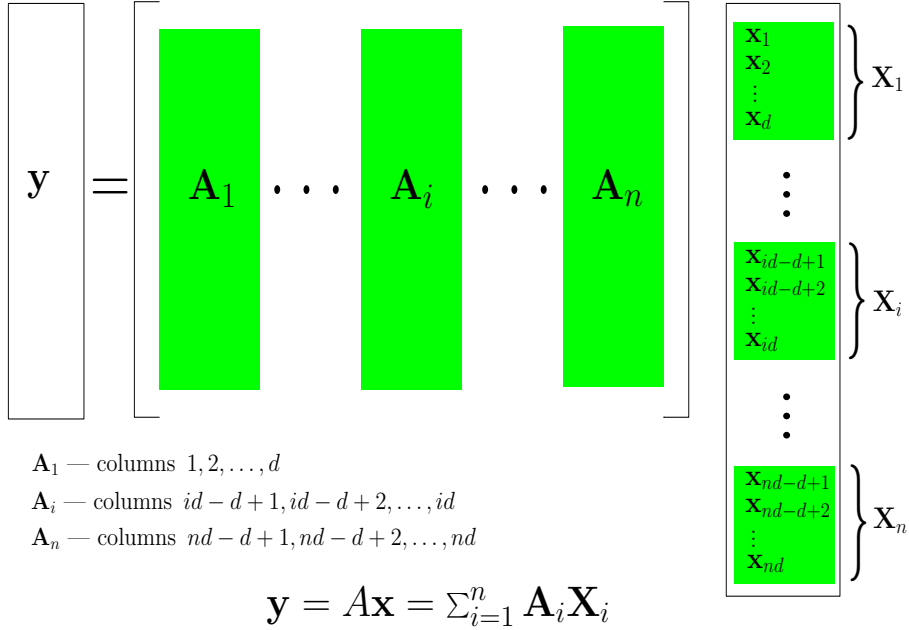,width=12cm,height=10cm}}
%%%%%%\end{minipage}
\vspace{-0.55in} \caption{Block-sparse model} \label{fig:blspmodel}
\end{figure}
Then we will call any signal $\x$ k-block-sparse if its at most
$k=\frac{K}{d}$ blocks $\X_i$ are non-zero (non-zero block is a block
that is not a zero block; zero block is a block that has all
elements equal to zero). Since $k$-block-sparse signals are
$K$-sparse one could then use (\ref{eq:l1}) to recover the solution
of (\ref{eq:system}). While this is possible, it clearly uses the
block structure of $\x$ in no way. To exploit the block structure of
$\x$ in \cite{SPH} the following polynomial-time algorithm (essentially a
combination of $\ell_2$ and $\ell_1$ optimizations) was considered (see also e.g. \cite{ZeWaSeGoAd08,ZeGoAd09,BCDH08,EKB09,BerFri09})
\begin{eqnarray}
\mbox{min} & & \sum_{i=1}^{n}\|\x_{(i-1)d+1:id}\|_2\nonumber \\
\mbox{subject to} & & A\x=\y. \label{eq:l2l1}\vspace{-.1in}
\end{eqnarray}
Extensive simulations in \cite{SPH} demonstrated that as $d$ grows
the algorithm in (\ref{eq:l2l1}) significantly outperforms the
standard $\ell_1$. The following was shown in \cite{SPH} as well:
let $A$ be an $M\times N$ matrix with a basis of null-space
comprised of i.i.d. Gaussian elements; if
$\alpha=\frac{M}{N}\rightarrow 1$ then there is a constant $d$ such
that all $k$-block-sparse signals $\x$ with sparsity $K\leq \beta N,
\beta\rightarrow \frac{1}{2}$, can be recovered with overwhelming
probability by solving (\ref{eq:l2l1}).
The precise relation between
$d$ and how fast $\alpha\longrightarrow 1$ and $\beta\longrightarrow
\frac{1}{2}$ was quantified in \cite{SPH} as well. In \cite{StojnicICASSP09block,StojnicJSTSP09} we extended the results from
\cite{SPH} and obtained the values of the recoverable block-sparsity for any
$\alpha$, i.e. for $0\leq \alpha \leq 1$. More precisely, for any
given constant $0\leq \alpha \leq 1$ we in \cite{StojnicICASSP09block,StojnicJSTSP09} determined a constant
$\beta=\frac{K}{N}$ such that for a sufficiently large $d$ (\ref{eq:l2l1})
with overwhelming probability
recovers any $k$-block-sparse signal with sparsity less then $K$.
(Under overwhelming probability we in this paper assume
a probability that is no more than a number exponentially decaying in $N$ away from $1$.)

%Clearly, for any given constant $\alpha\leq 1$ there is a maximum
%allowable value of the constant $\beta$ such that (\ref{eq:l2l1})
%finds solution of (\ref{eq:system}) with overwhelming probability
%for \emph{any} $K$-block-sparse $\x$. This maximum allowable value of the constant
%$\beta$ is called the \emph{strong threshold} (see
%\cite{DonohoUnsigned,DonohoPol}). We will denote the value of the strong
%threshold by $\beta_s$. Similarly, for any given constant
%$\alpha\leq 1$ one can define the \emph{sectional threshold} as the
%maximum allowable value of the constant $\beta$ such that
%(\ref{eq:l2l1}) finds the solution of (\ref{eq:system}) with overwhelming
%probability for \emph{any} $K$-block-sparse $\x$ with a given fixed location of non-zero blocks (see \cite{DonohoUnsigned,DonohoPol}).
%In a similar fashion one can then denote the value of the sectional threshold by $\beta_{sec}$. Finally, for any given constant
%$\alpha\leq 1$ one can define the \emph{weak threshold} as the
%maximum allowable value of the constant $\beta$ such that
%(\ref{eq:l2l1}) finds the solution of (\ref{eq:system}) with overwhelming
%probability for \emph{any} $K$-block-sparse $\x$ with a given fixed location of non-zero blocks and given fixed directions of non-zero block vectors $\X_i$ (see \cite{DonohoUnsigned,DonohoPol}).
%In a similar fashion one can then denote the value of the weak threshold by $\beta_{w}$.

Clearly, for any given constant $\alpha\leq 1$ there is a maximum
allowable value of $\beta$ such that for \emph{any} given $k$-sparse $\x$ in (\ref{eq:system}) the solution of (\ref{eq:l2l1})
is with overwhelming probability exactly that given $k$-sparse $\x$. We will refer to this maximum allowable value of
$\beta$ as the \emph{strong threshold} (see
\cite{DonohoPol,StojnicCSetamBlock09}). Similarly, for any given constant
$\alpha\leq 1$ and \emph{any} given $\x$ with a given fixed location and a given fixed directions of non-zero blocks
there will be a maximum allowable value of $\beta$ such that
(\ref{eq:l2l1}) finds that given $\x$ in (\ref{eq:system}) with overwhelming
probability. We will refer to this maximum allowable value of
$\beta$ as the \emph{weak threshold} and will denote it by $\beta_{w}$ (see, e.g. \cite{StojnicICASSP09,StojnicCSetam09}).

While \cite{StojnicICASSP09block,StojnicJSTSP09} provided fairly sharp strong threshold values they had done so in a somewhat asymptotic sense. Namely, the analysis presented in \cite{StojnicICASSP09block,StojnicJSTSP09} assumed fairly large values of block-length $d$. As such the analysis in \cite{StojnicICASSP09block,StojnicJSTSP09} then provided an ultimate performance limit of $\ell_2/\ell_1$-optimization rather than its performance characterization as a function of a particular fixed block-length.

In our own work \cite{StojnicCSetamBlock09} we extended the results of \cite{StojnicICASSP09block,StojnicJSTSP09} and provided a novel probabilistic framework for performance characterization of (\ref{eq:l2l1}) through which we were finally able to view block-length as a parameter of the system (the heart of the framework was actually introduced in \cite{StojnicCSetam09} and it seemed rather powerful; in fact, we afterwards found hardly any sparse type of problem that the framework was not able to handle with almost impeccable precision). Using the framework we obtained lower bounds on $\beta_w$. These lower bounds were in an excellent numerical agreement with the values obtained for $\beta_w$ through numerical simulations. One would therefore be tempted to believe that our lower bounds from \cite{StojnicCSetamBlock09} are tight. In this paper we design a mechanism that can be used to compute the upper bounds on $\beta_w$ (as it was the case with the framework of \cite{StojnicCSetamBlock09}, the new framework does not seem to be restricted in any way to the $\ell_2/\ell_1$ type of sparsity). The obtained upper bounds will match the lower bounds computed in \cite{StojnicCSetamBlock09} and essentially make them optimal. We should also point out that in our recent work \cite{StojnicUpper10} we created results similar in flavor to those that we will present here but are valid for general under-determined systems with sparse solutions (i.e. not necessarily those with block-sparse solutions). When viewed in that context the results presented here are a block analogue to those presented in \cite{StojnicUpper10}.

Before going through the details of our own approach we briefly take a look back and mention a few other known approaches from a vast literature cited above that have recently attracted significant amount of attention. The first thing one can think of when facing the block-structured unknown vectors is how to extend results known in the non-block (i.e. standard) case. In \cite{TGS05} the standard OMP (orthogonal matching pursuit) was generalized so that it can handle the jointly-sparse vectors more efficiently and improvements over the standard OMP were demonstrated. In \cite{BCDH08,EMsub} algorithms similar to the one from this paper were considered. It was explicitly shown through the block-RIP (block-restricted isometry property) type of analysis (which essentially extends to the block case the concepts introduced in \cite{CRT} for the non-block scenario) that one can achieve improvements in recoverable thresholds compared to the non-block case. Also, important results were obtained in \cite{EldRau09} where it was shown (also through the block-RIP type of analysis) that if one considers average case recovery of jointly-sparse signals the improvements in recoverable thresholds over the standard non-block signals are possible (of course, trivially, jointly-sparse recovery offers no improvement over the standard non-block scenario in the worst case). All these results provided a rather substantial basis for belief that the block-sparse recovery can provably be significantly more successful than the standard non-block one (as mentioned above, extensive simulations in \cite{SPH} confirmed such expectations). In \cite{StojnicCSetamBlock09} we then provided further results in this direction and here we establish their optimality.

We organize the rest of the paper in the following way. In Section
\ref{sec:theorems} we introduce two key theorems that will be the heart of our subsequent analysis.
In Section
\ref{sec:unsigned} we create the mechanism for computing the upper bounds on $\beta_w$.
Finally, in Section \ref{sec:discuss} we discuss obtained results.

%%%%%%%%%%%%%%%%%%%%%%%%%%%%%%%%%%%%%%%%%%%%%%%%%%%%%%%%%%%%%%%%%
\section{Key theorems}
\label{sec:theorems}
%%%%%%%%%%%%%%%%%%%%%%%%%%%%%%%%%%%%%%%%%%%%%%%%%%%%%%%%%%%%%%%%%

In this section we introduce two useful theorems that will be of key importance in our subsequent
analysis. First we recall on a null-space characterization of
$A$ that guarantees that the solution of
(\ref{eq:l2l1}) is the $k$-block-sparse solution of (\ref{eq:system}). Moreover, the characterization will establish this for any $\beta n$-block-sparse $\x$ with a fixed location and a fixed combination of directions of nonzero blocks. Since the analysis will clearly be irrelevant with respect to what particular location and what particular combination of directions of nonzero blocks are chosen, we can for the simplicity of the exposition and without loss of generality assume that the blocks $\X_{1},\X_{2},\dots,\X_{n-k}$ of $\x$ are equal to zero and the blocks $\X_{n-k+1},\X_{n-k+2},\dots,\X_n$ of $\X$ have fixed directions. Moreover, (as mentioned earlier) throughout the paper we will call such an $\x$ $k$-block-sparse. Under this assumption we have the following theorem from \cite{StojnicICASSP09block} that provides
such a characterization (similar characterizations that relate to the non-block case can be found in
\cite{DH01,FN,LN,Y,XHapp,SPH,DTbern}; furthermore, if instead of $\ell_1$ one, for
example, uses an $\ell_q$-optimization ($0<q<1$) in (\ref{eq:l2l1}) then
characterizations similar to the ones from
\cite{DH01,FN,LN,Y,XHapp,SPH,DTbern,GN03,GN04,GN07} can be derived as well.
\begin{theorem}(Nonzero part of $\x$ has fixed directions and location)
Assume that an $dm\times dn$ matrix $A$ is given. Let $\x$
be a $k$-block-sparse. Also let $\X_1=\X_2=\dots=\X_{n-k}=0$ and let the directions of vectors $\X_{n-k+1},\X_{n-k+2},\dots,\X_n$ be fixed.
Further, assume that $\y=A\x$ and that $\w$ is
an $dn\times 1$ vector with blocks $\W_i,i=1,\dots,n$, defined in a way analogous to the definition of blocks $\X_i$. If
\begin{equation}
(\forall \w\in \textbf{R}^{dn} | A\w=0) \quad  -\sum_{i=n-k+1}^n \frac{\X_i^T\W_i}{\|\X_i\|_2}<\sum_{i=1}^{n-k}\|\W_{i}\|_2.
\label{eq:thmeqgenweak1}
\end{equation}
then the solution of (\ref{eq:l2l1}) is $\x$. Moreover, if
\begin{equation}
(\exists \w\in \textbf{R}^{dn} | A\w=0) \quad  -\sum_{i=n-k+1}^n \frac{\X_i^T\W_i}{\|\X_i\|_2}>\sum_{i=1}^{n-k}\|\W_{i}\|_2.
\label{eq:thmeqgenweak2}
\end{equation}
then there will be a $k$-block-sparse $\x$ from the above defined set that satisfies (\ref{eq:system}) and is not the solution of (\ref{eq:l2l1}).
\label{thm:thmgenweak}
\end{theorem}
\begin{proof}
The first part follows directly from Corollary $2$ in \cite{StojnicCSetamBlock09}. The second part follows by combining (adjusting to the block case) the first part and the ideas of the second part of Theorem $1$ in \cite{StojnicUpper10}.
\end{proof}

Before proceeding further we would like to make a point similar to the one we have made in \cite{StojnicUpper10}. In our opinion the first part of the theorem that was put forth in \cite{StojnicCSetamBlock09} (and in essence in \cite{StojnicICASSP09}) is the unsung hero of all the success achieved in the thresholds analysis through various frameworks that we eventually designed. As mentioned in \cite{StojnicUpper10}, it was fist recognized in \cite{StojnicICASSP09} that characterizations of the type given in the first part of the above theorem could lead to the optimal threshold performance. As it became later clear the analysis in \cite{StojnicICASSP09} stopped somewhat short of the ultimate goal and it achieved only a moderate success in performance characterization of $\ell_1$-optimization. While the analysis of \cite{StojnicCSetam09} formally completed the task of evaluating fairly precisely the achievable thresholds it is the first part of the above theorem (or rather its a non-block equivalent from \cite{StojnicICASSP09}) that made everything possible. Along the same lines, while the framework created in \cite{StojnicCSetam09} was good enough to fairly precisely evaluate the achievable thresholds it is the first part of the above theorem that made the block generalization of results from \cite{StojnicCSetam09} possible in \cite{StojnicCSetamBlock09}.

Now, with regard to the second part of the above theorem, the story is of course similar. Its proof is rather simple and in fact almost completely follows the ideas of the non-block case (see, e.g. \cite{StojnicUpper10,DH01,GN03}). It is just that we never presented it before. Basically, we did not find the second part of the theorem to be of any (let alone much) use if one were to create the lower bounds on the thresholds. However, as the reader might guess, if one is concerned with proving the upper bounds the second part of the above theorem becomes the same type of the unsung hero that the first one was for the success of the framework of \cite{StojnicCSetamBlock09}. Below we use it to create a machinery as powerful as the one from \cite{StojnicCSetamBlock09} that provides the corresponding framework for upper-bounding the thresholds.

Before moving to the design of the framework, we would also like to say a few words about a possible design of the matrix $A$ that would satisfy the conditions of Theorem \ref{thm:thmgenweak}. Designing matrix $A$ such that
(\ref{eq:thmeqgenweak1}) holds would not be that hard.
The problem is that one does not know \emph{a priori} which $k$ blocks of $\x$ will be nonzero and which directions they will have. That would essentially force one to design $A$ such that (\ref{eq:thmeqgenweak1}) holds for any subset of $\{1,2,\dots,n\}$ of cardinality $k$ and any combination of directions on that subset. If one
assumes that $m$ and $k$ are proportional to $n$ (the case of our
interest in this paper) this is an enormous combinatorial task and the construction of such a deterministic
matrix $A$ is clearly not easy (in fact, as observed in e.g. \cite{StojnicCSetam09,StojnicCSetamBlock09} one may say that its a non-block counterpart is one of the most fundamental open problems in the area of theoretical compressed sensing; more on an equally important inverse problem of checking if a given matrix satisfies the condition of Theorem \ref{thm:thmgenweak} for any subset of $\{1,2,\dots,n\}$ of cardinality $k$ and any combination of block directions, the interested reader can find in \cite{JudNem08,DaEl08}). On the other hand, turning to
random matrices significantly simplifies things. As we will see later in the paper, Gaussian random matrices $A$ will turn out to be a very convenient choice.
The following phenomenal result from \cite{Gordon88} that relates to such matrices will be the key ingredient in the analysis that will follow.
\begin{theorem}(\cite{Gordon88})
\label{thm:Gordonmesh1} Let $X_{ij}$ and $Y_{ij}$, $1\leq i\leq n,1\leq j\leq m$, be two centered Gaussian processes which satisfy the following inequalities for all choices of indices
\begin{enumerate}
\item $E(X_{ij}^2)=E(Y_{ij}^2)$
\item $E(X_{ij}X_{ik})=E(Y_{ij}Y_{ik})$
\item $E(X_{ij}X_{lk})=E(Y_{ij}Y_{lk}), i\neq l$.
\end{enumerate}
Then
\begin{equation*}
P(\bigcap_{i}\bigcup_{j}(X_{ij}\geq \lambda_{ij}))\leq P(\bigcap_{i}\bigcup_{j}(Y_{ij}\geq \lambda_{ij})).
\end{equation*}
\end{theorem}

%%%%%%%%%%%%%%%%%%%%%%%%%%%%%%%%%%%%%%%%%%%%%%%%%%%%%%%%%%%
\section{Upper-bounding $\beta_w$ -- general $\x$} \label{sec:unsigned}
%%%%%%%%%%%%%%%%%%%%%%%%%%%%%%%%%%%%%%%%%%%%%%%%%%%%%%%%%%%

In this section we probabilistically analyze validity of the null-space characterization given in the second part of Theorem \ref{thm:thmgenweak}. Essentially, we will design a mechanism for computing upper bounds on $\beta_w$ (in fact, since it will be slightly more convenient we will actually determine lower bounds on $\alpha$; that is of course conceptually  the same as finding the upper-bounds on $\beta$).

We start by defining a quantity $\tau$ that will play one of the key roles below
\begin{eqnarray}
\tau(A) =  \min & & (\sum_{i=1}^{n-k}\|\W_i\|_2+\sum_{i=n-k+1}^{n}\frac{\X_i^T\W_i}{\|\X_i\|_2}) \nonumber \\
\mbox{subject to} & &  A\w=0\nonumber \\
& & \|\w\|_2\leq 1.\label{eq:deftau}
\end{eqnarray}
Now, we will in the rest of the paper assume that the entries of $A$ i.i.d. standard normal random variables. Then one can say that for any $\alpha$ and $\beta$ for which
\begin{equation}
\lim_{n\rightarrow\infty}P(\tau(A)<0)=1, \label{eq:mcrit}
\end{equation}
there is a $k$-block-sparse $\x$ (from a set of $\x$'s with a given fixed location of nonzero blocks and a given fixed combination of their directions) which (\ref{eq:l2l1}) with probability $1$ fails to find. For a fixed $\beta$ our goal will be to find the largest possible $\alpha$ for which (\ref{eq:mcrit}) holds, i.e. for which (\ref{eq:l2l1}) fails with probability $1$.

Before going through the randomness of the problem and evaluation of $P(\tau(A)<0)$ we will try to provide a more explicit expression for $\tau$ than the one given by the optimization problem in (\ref{eq:deftau}).
As a first step we write the Lagrange dual of (\ref{eq:deftau}) over $\w$
\begin{eqnarray}
\tau(A) =  \max_{\nu,\gamma}\min_{\w} & & (\sum_{i=1}^{n-k}\|\W_i\|_2+\sum_{i=n-k+1}^{n}\frac{\X_i^T\W_i}{\|\X_i\|_2})+\nu^T A\w +\gamma\sum_{i=1}^{n}\|\W_i\|^2-\gamma) \nonumber \\
\mbox{subject to} & &  \gamma\geq 0.\label{eq:deftau5}
\end{eqnarray}
To simplify the exposition we set
\begin{eqnarray}
\psi^{(i)} & = & \nu^T A_i,n-k+1\leq i\leq n-k\nonumber \\
\psi^{(i)} & = & \frac{\X_i}{\|\X_i\|_2}+\nu^T A_i,n-k+1\leq i\leq n,\label{eq:defpsi}
\end{eqnarray}
and assume that $\psi^{(i)}_j,1\leq j\leq d$, is the $j$-th component of $\psi^{(i)}$. Then one can rewrite (\ref{eq:deftau5}) in the following way
\begin{eqnarray}
\tau(A) =  \max_{\nu,\gamma}\min_{\w} & & (\sum_{i=1}^{n-k}\|\W_i\|_2+\sum_{i=1}^{n}\psi^{(i)}\W_i +\gamma\sum_{i=1}^{n}\|\W_i\|^2-\gamma) \nonumber \\
\mbox{subject to} & &  \gamma\geq 0.\label{eq:deftau6}
\end{eqnarray}
Let
\begin{equation}
f_1(\nu,\gamma,\w)=\sum_{i=1}^{n-k}\|\W_i\|_2+\sum_{i=1}^{n}\psi^{(i)}\W_i +\gamma\sum_{i=1}^{n}\|\W_i\|^2-\gamma.\label{eq:deff1}
\end{equation}
We then proceed by solving the inner minimization in (\ref{eq:deftau6}). Since $f_1(\cdot)$ is convex in $\w$ we simply find the optimal $\w$ by equaling the derivative of $f_1(\cdot)$ with respect to $\w$ to zero. We then have
\begin{eqnarray}
\frac{df_1(\nu,\gamma,\w)}{d\W_i} & = & \frac{\W_i^T}{\|\W_i\|_2}+\psi^{(i)} +2\gamma\W_i^T=0, 1\leq i\leq n-k\nonumber \\
\frac{df_1(\nu,\gamma,\w)}{d\W_i} & = & \psi_i +2\gamma\W_i=0, n-k+1\leq i\leq n,
\label{eq:deftau8}
\end{eqnarray}
where $0$'s are obviously $d$-dimensional row vectors of all zeros. At this point we will make an assumption that the above system can be solved. If it indeed can be solved the solution must satisfy
\begin{equation}
\W_i^T(2\gamma+\frac{1}{\|\W_i\|_2})=-\psi^{(i)}, 1\leq i\leq n-k,\label{eq:deftau9}
\end{equation}
and one would have
\begin{equation}
2\|\W_i\|_2\gamma+1=\|\psi^{(i)}\|_2, 1\leq i\leq n-k.\label{eq:deftau10}
\end{equation}
After plugging the value of $\|\W_i\|_2, 1\leq i\leq n-k$, back in (\ref{eq:deftau9}) we have
\begin{equation}
\W_i^T=-\psi^{(i)}\frac{\|\psi^{(i)}\|_2-1}{2\gamma\|\psi^{(i)}\|_2}, 1\leq i\leq n-k.\label{eq:deftau11}
\end{equation}
We should now note that for any $i\in\{1,\dots,n-k\}$, $\W_i$ from (\ref{eq:deftau11}) is indeed the solution of (\ref{eq:deftau8}) if $\|\psi^{(i)}\|_2\geq 1$. Otherwise one has $\W_i=0$ (here obviously $0$ stands for a column vector of $d$ zeros). On the other hand, from the second set of equations in (\ref{eq:deftau8}) one easily has
\begin{equation}
\W_i^T=-\frac{\psi^{(i)}}{2\gamma}, n-k+1\leq i\leq n.\label{eq:deftau12}
\end{equation}
Plugging the results from (\ref{eq:deftau11}) and (\ref{eq:deftau12}) back in (\ref{eq:deff1}) we obtain
\begin{multline}
\min_{\w}f_1(\nu,\gamma,\w)=\sum_{i=1}^{n-k}\frac{|\|\psi^{(i)}\|_2-1|_{\geq 0}}{2\gamma}-\sum_{i=1}^{n-k}\|\psi^{(i)}\|_2\frac{|\|\psi^{(i)}\|_2-1|_{\geq 0}}{2\gamma} +\sum_{i=1}^{n-k}\frac{(|\|\psi^{(i)}\|_2-1|_{\geq 0})^2}{4\gamma}\\-\sum_{i=n-k+1}^{n}\frac{\|\psi^{(i)}\|_2^2}{4\gamma}-\gamma.\label{eq:deftau13}
\end{multline}
where
\begin{equation}
|\|\psi^{(i)}\|_2-1|_{\geq 0}=\begin{cases}\|\psi^{(i)}\|_2-1 & if \quad \|\psi^{(i)}\|_2-1\geq 0\\
0 & otherwise\end{cases}.
\end{equation}
Transforming \ref{eq:deftau13} further we have
\begin{multline}
\min_{\w}f_1(\nu,\gamma,\w)=-\sum_{i=1}^{n-k}\frac{(|\|\psi^{(i)}\|_2-1|_{\geq 0})^2}{4\gamma}-\sum_{i=n-k+1}^{n}\frac{\|\psi^{(i)}\|_2^2}{4\gamma}-\gamma.\label{eq:deftau14}
\end{multline}
A combination of (\ref{eq:deftau6}) and (\ref{eq:deftau14}) gives
\begin{eqnarray}
\tau(A) =  \max_{\nu,\gamma} & & -\sum_{i=1}^{n-k}\frac{(|\|\psi^{(i)}\|_2-1|_{\geq 0})^2}{4\gamma}-\sum_{i=n-k+1}^{n}\frac{\|\psi^{(i)}\|_2^2}{4\gamma}-\gamma\nonumber \\
\mbox{subject to} & &  \gamma\geq 0.\label{eq:deftau15}
\end{eqnarray}
After solving over $\gamma$ we finally have
\begin{equation}
\tau(A)=-\min_{\nu} \sqrt{\sum_{i=1}^{n-k}(|\|\psi^{(i)}\|_2-1|_{\geq 0})^2+\sum_{i=n-k+1}^{n}\|\psi^{(i)}\|_2^2}.\label{eq:finaltau}
\end{equation}
The following rather small trick in rewriting the previous equation turns out to be useful
\begin{eqnarray}
\tau(A) = -\min_{\nu,\z} & & \sqrt{\sum_{i=1}^{n-k}(\|\psi^{(i)}\|_2-\z_i)^2+\sum_{i=n-k+1}^{n}\|\psi^{(i)}\|_2^2}\nonumber \\
\mbox{subject to} & & 0\leq \z_i\leq 1,1\leq i\leq n-k.\label{eq:finaltau1}
\end{eqnarray}
Now we set
\begin{equation}
f_2(\nu,A,\z)=\sqrt{\sum_{i=1}^{n-k}(\|\psi^{(i)}\|_2-\z_i)^2+\sum_{i=n-k+1}^{n}\|\psi^{(i)}\|_2^2}.\label{eq:deff2}
\end{equation}
Let $\Theta_i,n-k+1\leq i\leq n$, be $d\times d$ unitary matrices such that $\Theta_i\X_i=[1,0,0,0,\dots,0]^T$. Further let $Q^{(1)}\in R^{(n-k+dk)\times (dn)}$
be a matrix that has zeros everywhere except
\begin{eqnarray}
Q^{(1)}_{i,d(i-1)+1:di} & = & \q_i, 1\leq i\leq n-k\nonumber \\
Q^{(1)}_{d(i-1)+1:di,d(i-1)+1:di} & = & \Theta_i,n-k+1\leq i\leq n,\label{defQ1}
\end{eqnarray}
and let $\|\q_i\|_2=1,1\leq i\leq n-k$. One can then write
\begin{equation}
f_2(\nu,A,\z)=\max_{\a,Q^{(1)}}\a^T(Q^{(1)}(A^T\nu-\x^{(1)})-\z^{(1)}),\label{eq:f21}
\end{equation}
where $\x^{(1)}\in R^{dn}$, $\z^{(1)}\in R^{n-k+dk}$, $\a\in R^{n-k+dk}$, and
\begin{eqnarray}
\|\a\|_2 & = & 1\nonumber \\
\x^{(1)}_i & = & 0,1\leq i\leq n-k\nonumber \\
\x^{(1)}_i & = & \_i,d(n-k)+1\leq i\leq dn\nonumber \\
\z^{(1)}_i & = & \z_i,1\leq i\leq (n-k)\nonumber \\
\z^{(1)}_i & = & 0,n-k+1\leq i\leq n-k+dk.\label{deaxhzh}
\end{eqnarray}
Now let $Q\in R^{(n-k+dk)\times (dn)}$
be a matrix that has zeros everywhere except
\begin{eqnarray}
Q_{i,d(i-1)+1:di} & = & \q_i, 1\leq i\leq n-k\nonumber \\
Q_{d(n-k)+1:dn,d(n-k)+1:dn} & = & I,\label{defQ}
\end{eqnarray}
and again as above $\|\q_i\|_2=1,1\leq i\leq n-k$.
Furthermore, since $Q^{(1)}_{d(n-k)+1:dk,d(n-k)+1:dn}A_{d(n-k)+1:dn,1:dm}$ has the same distribution as $A_{d(n-k)+1:dn,1:dm}$ for the statistical purposes that we will consider later in the paper one can rewrite (\ref{eq:f21}) in the following way
\begin{equation}
f_2(\nu,A,\z)=\max_{\a}\a^T(Q(A^T\nu-\x^{(1)})-\z^{(1)}),\label{eq:f22}
\end{equation}
where $\z^{(1)}\in R^{dn}$ and $\a\in R^{n-k+dk}$ are as above and $\x^{(2)}\in R^{dn}$ has zeros everywhere except
\begin{eqnarray}
\x^{(2)}_{di+1}  =  1,n-k\leq i\leq n-1.\label{defx2}
\end{eqnarray}
At this point we are almost ready to switch to the probabilistic aspect of the analysis. To that end we do the last piece of transformation. Namely, we set $\z^{(2)}=\x^{(2)}+\z^{(1)}$ and rewrite (\ref{eq:finaltau}) as
\begin{eqnarray}
\tau(A)= -\min_{\z^{(2)},\nu} \max_{\|\a\|_2=1,\q_i} &  &
\a^T Q A^T \nu-\a^T\z^{(2)}  \nonumber \\
\mbox{subject to}& & 0\leq \z^{(2)}_i\leq 1, 1\leq i \leq n-k\nonumber \\
& & \z^{(2)}_{n-k+d(i-1)+2:n-k+di}=0_{1:d-1}, 1\leq i\leq k\nonumber \\
& & \z^{(2)}_{n-k+d(i-1)+1}=1, 1\leq i\leq k\nonumber \\
& & \|\q_i\|_2=1,1\leq i\leq n-k
\label{eq:finaldet}
\end{eqnarray}
where $0_{1:d-1}$ is a vector of $d-1$ zeros. Also, we will call $Z$ the set of all $\z^{(2)}$ that are feasible in (\ref{eq:finaldet}).
Now we are ready to invoke the results from Theorem \ref{thm:Gordonmesh1}. We do so through the following modification of the corresponding lemma from \cite{StojnicUpper10} which itself is a slightly modified version of Lemma 3.1 from \cite{Gordon88} (Lemma 3.1 is of course a direct consequence of Theorem \ref{thm:Gordonmesh1} and the backbone of the escape through a mesh theorem utilized in \cite{StojnicCSetam09}).
\begin{lemma}
Let $A$ be an $dm\times dn$ matrix with i.i.d. standard normal components. Let $\g$ and $\h$ be $dn\times 1$ and $dm\times 1$ vectors, respectively, with i.i.d. standard normal components. Also, let $g$ be a standard normal random variable and let $Z$ and $Q$ be as defined above. Then
\begin{equation}
\hspace{-.9in}P(\min_{\z^{(2)}\in Z,\nu\in R^{dm}\setminus 0}\max_{\|\a\|_2=1,Q}(\a^TQ A^T\nu +\|\nu\|_2 g-\zeta_{\a,\z^{(2)},\nu})\geq 0)\geq P(\min_{\z^{(2)}\in Z,\nu\in R^{dm}\setminus 0}\max_{\|\a\|_2=1,Q}(\|\nu\|_2\a^TQ\g+\sum_{i=1}^{dm}\h_i\nu_i-\zeta_{\a,\z^{(2)},\nu})\geq 0).\label{eq:problemma}
\end{equation}\label{eq:unsignedlemma}
\end{lemma}
\begin{proof}
The proof is exactly the same as the one of the corresponding lemma from \cite{StojnicUpper10} (or for that matter as the one of Lemma 3.1 in \cite{Gordon88}). The only difference is that in current context $\a^TQ$ plays the role that $\a^T$ played in the corresponding lemma in \cite{StojnicUpper10}.
\end{proof}
Let $\zeta_{\a,\z^{(2)},\nu}=\epsilon_{5}^{(g)}\sqrt{dn}\|\nu\|_2+\a^T\z^{(2)}$ with $\epsilon_{5}^{(g)}>0$ being an arbitrarily small constant independent of $n$. Then the left-hand side of the inequality in (\ref{eq:problemma}) is then the following probability of interest
\begin{equation*}
P(\min_{\z^{(2)}\in Z,\nu\in R^{dm}\setminus 0}\max_{\|\a\|_2=1,Q}(\|\nu\|_2\a^TQ\g+\sum_{i=1}^{dm}\h_i\nu_i-\epsilon_{5}^{(g)}\sqrt{dn}\|\nu\|_2-\a^T\z^{(2)})\geq 0).
\end{equation*}
After solving the inner maximization over $\a$ and $Q$ and pulling out $\|\nu\|_2$ one has
\begin{equation*}
P(\min_{\z^{(2)}\in Z,\nu\in R^{dm}\setminus 0}(\|\bar{\g}-\frac{1}{\|\nu\|_2}\z^{(2)}\|_2+\sum_{i=1}^{dm}\h_i\frac{\nu_i}{\|\nu\|_2}-\epsilon_{5}^{(g)}\sqrt{dn})\geq 0),
\end{equation*}
 where $\bar{\g}=[\g_{(1)},\g_{(2)},\dots,\g_{(n-k)},\g_{d(n-k)+1},\g_{d(n-k)+2},\dots,\g_{dn}]^T$, where $[\g_{(1)},\g_{(2)},\dots,\g_{(n-k)}]$ are magnitudes of vectors  $[\g_{1:d},\g_{d+1:2d},\dots,\g_{d(n-k)+1:d(n-k)}]$ sorted in increasing order. Minimization of the second term then gives us
\begin{equation}
P(\min_{\z^{(2)}\in Z,\nu\in R^{dm}\setminus 0}(\|\bar{\g}-\frac{1}{\|\nu\|_2}\z^{(2)}\|_2)\geq \|\h\|_2+\epsilon_{5}^{(g)}\sqrt{dn}).\label{eq:probanal1}
\end{equation}
Since $\h$ is a vector of $dm$ i.i.d. standard normal variables it is rather trivial that $P(\|\h\|_2<(1+\epsilon_{1}^{(m)})\sqrt{dm})\geq 1-e^{-\epsilon_{2}^{(m)} dm}$ where $\epsilon_{1}^{(m)}>0$ is an arbitrarily small constant and $\epsilon_{2}^{(m)}$ is a constant dependent on $\epsilon_{1}^{(m)}$ but independent of $n$. Then from (\ref{eq:probanal1}) one obtains
\begin{multline}
P(\min_{\z^{(2)}\in Z,\nu\in R^{dm}\setminus 0}(\|\bar{\g}-\frac{1}{\|\nu\|_2}\z^{(2)}\|_2)\geq \|\h\|_2+\epsilon_{5}^{(g)}\sqrt{dn})\\\geq (1-e^{-\epsilon_{2}^{(m)} dm})
P(\min_{\z^{(2)}\in Z,\nu\in R^{dm}\setminus 0}(\|\bar{\g}-\frac{1}{\|\nu\|_2}\z^{(2)}\|_2)\geq (1+\epsilon_{1}^{(m)})\sqrt{dm}+\epsilon_{5}^{(g)}\sqrt{dn})).\label{eq:probanal2}
\end{multline}
To make results as parallel as possible to the ones created in \cite{StojnicCSetamBlock09} we will now set $\G_i^*=\g_{d(n-k+i-1)+2:d(n-k+i)},\\1\leq i\leq k$ and
\begin{equation}
\bar{\G}=[\g_{(1)},\g_{(2)},\dots,\g_{(n-k)},\g_{d(n-k)+1},\g_{d(n-k+1)+1},\dots,\g_{d(n-1)+1},\|\G_1^*\|_2,\|\G_2^*\|_2,\dots,\|\G_k^*\|_2,]^T.\label{eq:defG}
\end{equation}
Moreover, we will set
\begin{equation}
Z_{\G}=\{\z^{(2)}|0\leq \z^{(2)}_i\leq 1,1\leq i\leq n-k,\z^{(2)}_i=1,n-k+1\leq i\leq n, \z^{(2)}_i=0,n+1\leq i\leq n+k\}.
\end{equation}
One can then rewrite (\ref{eq:probanal2}) in the following way
\begin{multline}
P(\min_{\z^{(2)}\in Z_{\G},\nu\in R^{dm}\setminus 0}(\|\bar{\G}-\frac{1}{\|\nu\|_2}\z^{(2)}\|_2)\geq \|\h\|_2+\epsilon_{5}^{(g)}\sqrt{dn})\\\geq (1-e^{-\epsilon_{2}^{(m)} dm})
P(\min_{\z^{(2)}\in Z_{\G},\nu\in R^{dm}\setminus 0}(\|\bar{\G}-\frac{1}{\|\nu\|_2}\z^{(2)}\|_2)\geq (1+\epsilon_{1}^{(m)})\sqrt{dm}+\epsilon_{5}^{(g)}\sqrt{dn})).\label{eq:probanal21}
\end{multline}

The optimization on the right-hand side of (\ref{eq:probanal2}) is structurally the same as the one in equation $(16)$ in \cite{StojnicCSetamBlock09} (actually to be more precise it is the same as the weak threshold equivalent to $(16)$). Essentially, the exact equivalence between these optimizations is achieved after in $(16)$ from \cite{StojnicCSetamBlock09}   $\tilde{\H}$ is replaced by $\bar{\G}$, $\nu$ is replaced by $\frac{1}{\|\nu\|_2}$, $\lambda$ is restricted to the lower $(n-k)$ components, and after one additionally notes that in $(16)$ from \cite{StojnicCSetamBlock09} $0\leq\lambda\leq \nu$, which corresponds to $0\leq \z^{(2)}_i\leq 1, 1\leq i\leq n-k$ introduced above (that way one would in essence obtain the weak threshold equivalent to $(16)$; this was not explicitly written anywhere in \cite{StojnicCSetamBlock09} but is rather obvious; in \cite{StojnicCSetamBlock09} we, instead, made a ``weak" equivalence to its $(30)$). With these replacements one can then use the machinery of \cite{StojnicCSetamBlock09} to establish
\begin{equation}
\min_{\z^{(2)}\in Z_{\G},\nu\in R^{dm}\setminus 0}(\|\bar{\G}-\frac{1}{\|\nu\|_2}\z^{(2)}\|_2)=\sqrt{\sum_{i=\cweak+1}^{n+k}\bar{\G}_i^2-\frac{((\bar{\G}
^T\z^{(2)})-\sum_{i=1}^{\cweak}\bar{\G}_i)^2}{n-\cweak
}}=\sqrt{f_{\G}(\cweak)}\label{eq:probanal4}
\end{equation}
where $\cweak$ is the solution of
\begin{equation}
\frac{(\bar{\G}
^T\z^{(2)})-\sum_{i=1}^{\cweak}\bar{\G}_i}{n-\cweak
}=\bar{\G}_{\cweak}.\label{eq:defcw}
\end{equation}
As a side remark, we should point out that the key point to the success of our method is that the derivation of \cite{StojnicCSetamBlock09}
establishes the equality in (\ref{eq:probanal4}). It is just that in \cite{StojnicCSetamBlock09} only the ``smaller than" inequality part of this equality was utilized. At this point we have established the core of our upper-bounding arguments. The rest is just a slightly modified repetition of the derivations from \cite{StojnicCSetamBlock09} (or one may think of them as a block parallelization of the derivations presented in \cite{StojnicUpper10}) so that we can make everything precise.

First we will define two quantities $c_{w}^{(l)}$ and $c_{w}^{(u)}$ as the solutions of the following two equations:
\begin{eqnarray}
\frac{(1-\epsilon_1^{(c)})E((\bar{\G}^T\z^{(2)})-\sum_{i=1}^{\cweak^{(l)}} \bar{\G}_i)}{n-\cweak^{(l)}}-F_a^{-1}\left (\frac{(1+\epsilon_1^{(c)}\cweak^{(l)}}{n(1-\beta_w)}\right ) & = & 0\nonumber \\
\frac{(1+\epsilon_2^{(c)})E((\bar{\G}^T\z^{(2)})-\sum_{i=1}^{\cweak^{(u)}} \bar{\G}_i)}{n-\cweak^{(u)}}-F_a^{-1}\left (\frac{(1-\epsilon_2^{(c)}\cweak^{(u)}}{n(1-\beta_w)}\right ) & = & 0.\label{eq:probanal44}
\end{eqnarray}
where $F_a^{-1}(\cdot)$ is the inverse cdf of the chi random variable with $d$ degrees of freedom ($\chi_d$), and $\epsilon_{i}^{(c)}>0,1\leq i\leq 2$ are arbitrarily small constants independent of $n$. It follows then directly from the derivation $(33)-(44)$ in \cite{StojnicCSetamBlock09} that
\begin{equation}
P(c_w\in\{c_w^{(l)},c_w^{(u)}\})\geq 1-e^{-\epsilon_{3}^{(c)}n}\label{eq:probcw}
\end{equation}
where $\epsilon_{3}^{(c)}$ is a constant dependent on $\epsilon_{i}^{(c)}>0,1\leq i\leq 2$, $\cweak^{(l)}$, $\cweak^{(u)}$ but independent of $n$.
We now set $c_w=c_w^{(u)}$ and focus on (\ref{eq:probanal4}). Concentration analysis machinery of \cite{StojnicCSetamBlock09} will help us establish a ``high probability" lower bound on $f_{\G}(c_w)$ (this will amount to nothing but reversing the concentration arguments that we have established in \cite{StojnicCSetamBlock09}; concentration arguments are of course easy to reverse; what was harder to reverse was the part before (\ref{eq:probanal4})). We now split $f_{\G}(c_w)$ into two parts i.e.
\begin{equation}
f_{\G}(c_w)=f_{\G}^{(1)}(\cweak)-f_{\G}^{(2)}(\cweak),\label{eq:deff}
\end{equation}
where $f_{\G}^{(1)}(\cweak)=\sum_{i=\cweak+1}^n\bar{\G}_i^2$ and $f_{\G}^{(2)}(\cweak)=(\bar{\G}
^T\z^{(2)})-\sum_{i=1}^{\cweak}\bar{\G}_i)$. Now, $f_{\G}^{(1)}(\cweak)$ concentrates trivially, the argument is the same as the one that can be established when $\cweak=0$ (alternatively one can repeat derivation $(42)$ from \cite{StojnicCSetam09} to obtain the Lipschitz constant and combine it with Lipschitz concentration formula $(36)$ also in \cite{StojnicCSetamBlock09}). So we have
\begin{equation}
P(f_{\G}^{(1)}(\cweak)\geq (1-\epsilon_{1}^{(g)})Ef_{\G}^{(1)}(\cweak))>1-e^{-\epsilon_{2}^{(g)}n},\label{eq:concf1}
\end{equation}
again as usual $\epsilon_{1}^{(g)}>0$ is an arbitrarily small constant and $\epsilon_{2}^{(g)}$ is a constant dependent on $\epsilon_{1}^{(g)}$ and $\cweak$ but independent of $n$. On the other hand, concentration of $f_{\G}^{(2)}(\cweak)$ follows by reversing $(43)$ from \cite{StojnicCSetamBlock09}, i.e.
\begin{equation}
P(f_{\G}^{(2)}(\cweak)\geq (1+\epsilon_{3}^{(g)})Ef_{\G}^{(2)}(\cweak))>1-e^{-\epsilon_{4}^{(g)}n}\label{eq:concf2}
\end{equation}
where again as usual $\epsilon_{3}^{(g)}>0$ is an arbitrarily small constant and $\epsilon_{4}^{(g)}$ is a constant dependent on $\epsilon_{3}^{(g)}$ and $\cweak$ but independent of $n$. Combination of (\ref{eq:probanal4}), (\ref{eq:concf1}) and (\ref{eq:concf1}) gives (the only other thing one should observe here is that $E((\bar{\G}
^T\z^{(2)})-\sum_{i=1}^{\cweak}\bar{\G}_i)\geq 0$)
\begin{multline}
\hspace{-.4in}P\left (\sqrt{\sum_{i=\cweak+1}^{n+k}\bar{\G}_i^2-\frac{((\bar{\G}
^T\z^{(2)})-\sum_{i=1}^{\cweak}\bar{\G}_i)^2}{n-\cweak
}}\geq \sqrt{(1-\epsilon_{1}^{(g)})E\sum_{i=\cweak+1}^{n+k}\bar{\G}_i^2-\frac{(1+\epsilon_{3}^{(g)})^2(E((\bar{\G}
^T\z^{(2)})-\sum_{i=1}^{\cweak}\bar{\G}_i))^2}{n-\cweak
}}\right )\\\geq (1-e^{-\epsilon_{2}^{(g)}n})(1-e^{-\epsilon_{4}^{(g)}n}).\label{eq:probanal5}
\end{multline}
%\begin{equation}
%P(\sqrt{\sum_{i=\cweak+1}^n\bar{\G}_i^2-\frac{((\bar{\G}
%^T\z)-\sum_{i=1}^{\cweak}\bar{\G}_i)^2}{n-\cweak
%}}\geq (1-\epsilon_5)\sqrt{n-\cweak})\geq 1-e^{-\epsilon_6 n}\label{eq:probanal5}
%\end{equation}
%where $\epsilon_5>0$ is an arbitrarily small constant and $\epsilon_6$ is a constant dependent on $\epsilon_5$ but independent of $n$. Establishing (\ref{eq:probanal5}) consists of two steps. In the first step one should show that $\sum_{i=\cweak+1}^n\bar{\G}_i^2-\frac{((\bar{\G}
%^T\z)-\sum_{i=1}^{\cweak}\bar{\G}_i)^2}{n-\cweak}$ concentrates around its mean. That is an easy but somewhat tedious exercise that involves reapplication of the derivation leading to $(38)$ in \cite{StojnicCSetam09} and we omit it. The second step then requires one to show that  $E(\sum_{i=\cweak+1}^n\bar{\G}_i^2-\frac{((\bar{\G}
%^T\z)-\sum_{i=1}^{\cweak}\bar{\G}_i)^2}{n-\cweak})=n-\cweak$ which is what is done in the weak threshold theorem in \cite{StojnicCSetam09}.
%
Now, let
\begin{equation}
dm_{w}=\frac{1}{(1+\epsilon_{1}^{(m)})^2}(\sqrt{(1-\epsilon_{1}^{(g)})E\sum_{i=\cweak+1}^{n+k}\bar{\G}_i^2-\frac{(1+\epsilon_{3}^{(g)})^2(E((\bar{\G}
^T\z^{(2)})-\sum_{i=1}^{\cweak}\bar{\G}_i))^2}{n-\cweak
}}-\epsilon_{3}^{(m)}\sqrt{dn}-\epsilon\sqrt{dn})^2,\label{eq:defmmax}
\end{equation}
where $\epsilon_{3}^{(m)}>0$ is an arbitrarily small constant. Combining (\ref{eq:probanal2}), (\ref{eq:probanal4}), (\ref{eq:probcw}), (\ref{eq:probanal5}), and  (\ref{eq:defmmax}) we have
\begin{equation}
P(\min_{\z^{(2)}\in Z_{\G},\nu\in R^{dm}\setminus 0}(\|\bar{\G}-\frac{1}{\|\nu\|_2}\z^{(2)}\|_2)\geq \|\h\|_2+\epsilon_{5}^{(g)}\sqrt{dn})\\ \geq (1-e^{-\epsilon_{2}^{(m)}m})
(1-e^{-\epsilon_{2}^{(g)} n})(1-e^{-\epsilon_{4}^{(g)} n})(1-e^{-\epsilon_{3}^{(c)}n}).\label{eq:probanal6}
\end{equation}
Further combination of (\ref{eq:problemma}), (\ref{eq:probanal1}), (\ref{eq:probanal2}), and (\ref{eq:probanal6}) gives us that if $m=m_{w}$
\begin{equation}
\hspace{-.7in}P(\min_{\z^{(2)}\in Z_{\G},\nu\in R^{dm}\setminus 0}\max_{\|\a\|_2=1,Q}(\a^T Q A^T\nu-\a^T\z^{(2)}+\|\nu\|_2(g-\epsilon_{5}^{(g)}\sqrt{dn}))\geq 0)\geq (1-e^{-\epsilon_{2}^{(m)}dm_w})
(1-e^{-\epsilon_{2}^{(g)} n})(1-e^{-\epsilon_{4}^{(g)} n})(1-e^{-\epsilon_{3}^{(c)}n}).\label{eq:probanal7}
\end{equation}
Since $P(g\leq\epsilon_{5}^{(g)}\sqrt{dn})\geq 1-e^{-\epsilon_{6}^{(g)} dn}$ (where $\epsilon_{6}^{(g)}$ is, as all other $\epsilon$'s in this paper are, independent of $n$) from (\ref{eq:probanal7}) we finally have
\begin{equation}
\hspace{-.5in}P(\min_{\z^{(2)}\in Z_{\G},\nu\in R^{dm}\setminus 0}\max_{\|\a\|_2=1,Q}(-\a^T QA^T\nu+\a^T\z^{(2)})> 0)\geq (1-e^{-\epsilon_{2}^{(m)}dm_w})
(1-e^{-\epsilon_{2}^{(g)} n})(1-e^{-\epsilon_{4}^{(g)} n})(1-e^{-\epsilon_{6}^{(g)} dn})(1-e^{-\epsilon_{3}^{(c)}n}).\label{eq:probanal8}
\end{equation}
Connecting (\ref{eq:finaldet}) and (\ref{eq:probanal8}) we obtain
\begin{equation*}
P(-\tau(A)> 0)\geq (1-e^{-\epsilon_{2}^{(m)}dm_w})
(1-e^{-\epsilon_{2}^{(g)} n})(1-e^{-\epsilon_{4}^{(g)} n})(1-e^{-\epsilon_{6}^{(g)} dn})(1-e^{-\epsilon_{3}^{(c)}n}),
\end{equation*}
and ultimately
\begin{equation}
\lim_{n\rightarrow\infty}P(\tau(A)< 0)=\lim_{n\rightarrow\infty} (1-e^{-\epsilon_{2}^{(m)}dm_w})
(1-e^{-\epsilon_{2}^{(g)} n})(1-e^{-\epsilon_{4}^{(g)} n})(1-e^{-\epsilon_{6}^{(g)} dn})(1-e^{-\epsilon_{3}^{(c)}n})=1\label{eq:finaltau}
\end{equation}
which is what we established as a goal in (\ref{eq:mcrit}).
We summarize the results in the following theorem.

\begin{theorem}(Exact weak threshold)
Let $A$ be a $dm\times dn$ measurement matrix in (\ref{eq:system})
with the null-space uniformly distributed in the Grassmanian. Let
the unknown $\x$ in (\ref{eq:system}) be $k$-block-sparse with the length of its blocks $d$. Further, let the location and the directions of nonzero blocks of $\x$ be arbitrarily chosen but fixed.
Let $k,m,n$ be large
and let $\alpha=\frac{m}{n}$ and $\betaweak=\frac{k}{n}$ be constants
independent of $m$ and $n$. Let $\gammainc(\cdot,\cdot)$ and $\gammainc^{-1}(\cdot,\cdot)$ be the incomplete gamma function and its inverse, respectively. Further,
let all $\epsilon$'s below be arbitrarily small constants.

\begin{enumerate}
\item Let $\hthetaweak$, ($\betaweak\leq \hthetaweak\leq 1$) be the solution of
\begin{equation}
(1-\epsilon_1^{(c)})(1-\betaweak)\frac{\frac{\sqrt{2}\Gamma(\frac{d+1}{2})}{\Gamma(\frac{d}{2})}
\left (1-\gammainc(\gammainc^{-1}(\frac{1-\thetaweak}{1-\betaweak},\frac{d}{2}),\frac{d+1}{2})\right )}{\thetaweak}-\sqrt{2\gammainc^{-1}(\frac{(1+\epsilon_1^{(c)})(1-\thetaweak)}{1-\betaweak},\frac{d}{2})}=0
.\label{eq:thmweaktheta}
\end{equation}
If $\alpha$ and $\betaweak$ further satisfy
\begin{multline}
\alpha d>(1-\betaweak)\frac{2\Gamma(\frac{d+2}{2})}{\Gamma(\frac{d}{2})}
\left (1-\gammainc(\gammainc^{-1}(\frac{1-\hthetaweak}{1-\betaweak},\frac{d}{2}),\frac{d+2}{2})\right )
+\betaweak d\\-\frac{\left ((1-\betaweak)\frac{\sqrt{2}\Gamma(\frac{d+1}{2})}{\Gamma(\frac{d}{2})}
(1-\gammainc(\gammainc^{-1}(\frac{1-\hthetaweak}{1-\betaweak},\frac{d}{2}),\frac{d+1}{2}))\right ) ^2}{\hthetaweak}\label{eq:thmweakalpha}
\end{multline}
then with overwhelming probability the solution of (\ref{eq:l2l1}) is the $k$-block-sparse $\x$ from (\ref{eq:system}).

\item Let $\htheta_w$, ($\beta_w\leq \htheta_w\leq 1$) be the solution of
\begin{equation}
(1+\epsilon_2^{(c)})(1-\betaweak)\frac{\frac{\sqrt{2}\Gamma(\frac{d+1}{2})}{\Gamma(\frac{d}{2})}
\left (1-\gammainc(\gammainc^{-1}(\frac{1-\thetaweak}{1-\betaweak},\frac{d}{2}),\frac{d+1}{2})\right )}{\thetaweak}-\sqrt{2\gammainc^{-1}(\frac{(1-\epsilon_2^{(c)})(1-\thetaweak)}{1-\betaweak},\frac{d}{2})}=0
.\label{eq:thmweaktheta1}
\end{equation}
If $\alpha$ and $\betaweak$ further satisfy
\begin{multline}
\alpha d<\frac{1}{(1+\epsilon_1^{(m)})^2}((1-\epsilon_1^{(g)}) (1-\betaweak)\frac{2\Gamma(\frac{d+2}{2})}{\Gamma(\frac{d}{2})}
\left (1-\gammainc(\gammainc^{-1}(\frac{1-\hthetaweak}{1-\betaweak},\frac{d}{2}),\frac{d+2}{2})\right )
+\betaweak d\\-\frac{\left ((1-\betaweak)\frac{\sqrt{2}\Gamma(\frac{d+1}{2})}{\Gamma(\frac{d}{2})}
(1-\gammainc(\gammainc^{-1}(\frac{1-\hthetaweak}{1-\betaweak},\frac{d}{2}),\frac{d+1}{2}))\right ) ^2}{\hthetaweak (1+\epsilon_3^{(g)})^{-2}} )\label{eq:thmweakalpha1}
\end{multline}
then with overwhelming probability there will be a $k$-block-sparse $\x$ (from a set of $\x$'s with fixed locations and directions of nonzero blocks) that satisfies (\ref{eq:system}) and is \textbf{not} the solution of (\ref{eq:l2l1}).
\end{enumerate}
\label{thm:thmweakthr}
\end{theorem}
\begin{proof}
The first part was established in \cite{StojnicCSetamBlock09}. The second part follows from the previous discussion combining (\ref{eq:thmeqgenweak2}), (\ref{eq:deftau}), (\ref{eq:mcrit}), (\ref{eq:probanal44}), (\ref{eq:defmmax}),  and (\ref{eq:finaltau}).
\end{proof}

%\noindent While the previous theorem insists on precision one can do what we will refer to as the ``deepsilonification" and obtain a way more convenient characterization. After removing all $\epsilon$'s (or say after setting them to the values that are so small when compared to $n$ that on any available finite precision machine they don't impact the above characterization) one in a more informal language then has.
%
%\noindent Assume the setup of the above theorem. Let $\alpha_w$ and $\beta_w$ satisfy the following:
%
%\noindent \underline{\underline{\textbf{Fundamental characterization of the $\ell_1$ performance:}}}
%
%\begin{center}
%\shadowbox{$
%%\begin{equation}
%(1-\beta_w)\frac{\sqrt{\frac{2}{\pi}}e^{-(\erfinv(\frac{1-\alpha_w}{1-\beta_w}))^2}}{\alpha_w}-\sqrt{2}\erfinv (\frac{1-\alpha_w}{1-\beta_w})=0.
%%\end{equation}
%$}
%-\vspace{-.5in}\begin{equation}
%\label{eq:thmweaktheta2}
%\end{equation}
%\end{center}
%
%
%Then:
%\begin{enumerate}
%\item If $\alpha>\alpha_w$ then with overwhelming probability the solution of (\ref{eq:l1}) is the $k$-sparse $\x$ from (\ref{eq:system}).
%\item If $\alpha<\alpha_w$ then with overwhelming probability there will be a $k$-sparse $\x$ (from a set of $\x$'s with fixed locations and signs of nonzero components) that satisfies (\ref{eq:system}) and is \textbf{not} the solution of (\ref{eq:l1}).
%    \end{enumerate}

The above theorem establishes the fundamental characterization of the $\ell_2/\ell_1$ performance. Numerical values of the weak threshold obtained using (\ref{eq:thmweaktheta}) and (\ref{eq:thmweakalpha}) were presented in \cite{StojnicCSetamBlock09}. As it was demonstrated there, the lower bounds on the thresholds were in an excellent numerical agreement with the recovery thresholds that can be obtained through numerical simulations. Theorem  \ref{thm:thmweakthr} establishes that the lower bounds computed in \cite{StojnicCSetamBlock09} (essentially those one can compute from (\ref{eq:thmweaktheta}) and (\ref{eq:thmweakalpha})) are actually the upper bounds as well and as such are the exact values of the weak thresholds.

%Moreover, in a companion paper \cite{StojnicEquiv10} we established a qualitative equivalence of the characterization given in (\ref{eq:thmweaktheta2}) and the results obtained in \cite{DonohoPol}. It is rather fascinating to us how well the axiomatic system of mathematics works and that two so seemingly different approaches, the geometric one from \cite{DonohoPol} and the purely probabilistic one from \cite{StojnicCSetam09}, result in exactly the same optimal characterization of the performance of $\ell_1$-optimization.

For the completeness we present in Figure \ref{fig:weakthr} again the plot obtained based on the ultimate characterization (\ref{eq:thmweaktheta}), (\ref{eq:thmweakalpha}).
\begin{figure}[htb]
%%%%%\begin{minipage}[b]{1.0\linewidth}
\centering
\centerline{\epsfig{figure=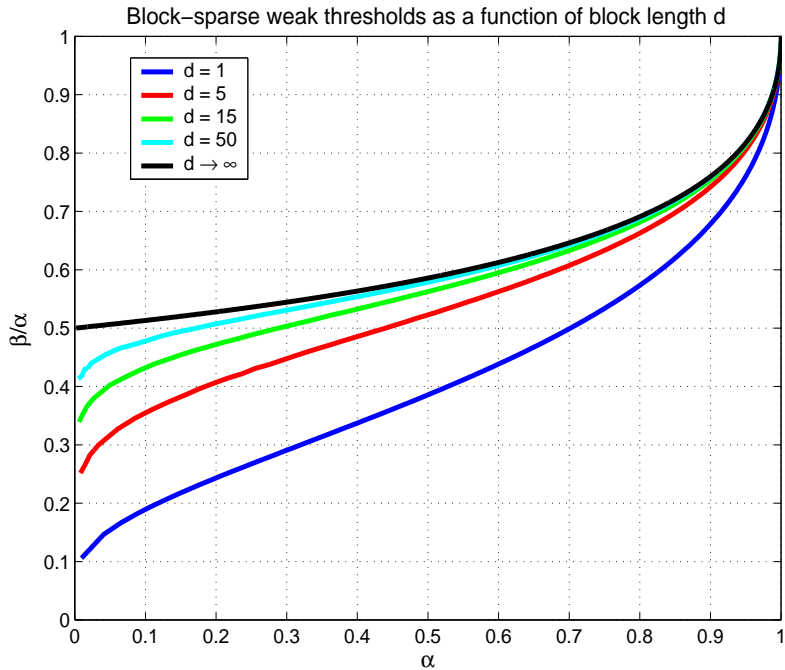,width=10.5cm,height=9cm}}
%%%%%%\end{minipage}
\vspace{-0.2in} \caption{\emph{Weak} threshold, $\ell_2/\ell_1$-optimization --- ultimate performance}
\label{fig:weakthr}
\end{figure}

\section{Discussion}
\label{sec:discuss}
%%%%%%%%%%%%%%%%%%%%%%%%%%%%%%%%%%%%%%%%%%%%%%%%%%%%%%%%%%%%%%%%%%%%%%%%%%%%%%%%

In this paper we considered under-determined linear systems of equations with sparse solutions. In our recent papers we created mechanisms that can be used to analyze almost to perfection the performance of a technique called $\ell_1$-optimization when used for solving such systems. When presenting those results we have mentioned that various generalizations are possible. In this paper we presented a set of such generalizations. The results that we presented here relate to a specific type of sparse vectors, namely the so-called block-sparse vectors.

We looked from a theoretical point of view at a classical polynomial-time
$\ell_2/\ell_1$-optimization algorithm that can be used for recovery of such vectors. Such an optimization algorithm is a natural generalization of the above mentioned $\ell_1$-optimization that is typically employed when the unknown vectors besides being sparse are not known to possess any other type of structure. Under the assumption that the system matrix $A$ has i.i.d. standard normal components,
we derived upper bounds on the values of the recoverable weak
thresholds in the so-called linear regime, i.e. in the regime when
the recoverable sparsity is proportional to the length of the
unknown vector. Obtained upper bounds match the corresponding lower bounds we found through a framework designed in \cite{StojnicCSetamBlock09}. A combination of the mechanism from \cite{StojnicCSetamBlock09} and the one that we presented in this paper is then enough to provide an explicit ultimate characterization of the success of $\ell_2/\ell_1$-optimization when applied in solving under-determined systems of linear equations with block-sparse solutions.

As mentioned in a companion paper \cite{StojnicUpper10}, further developments are then pretty much unlimited. Various specific problems that have been of interest in a broad scientific literature developed over the last few years can then easily be handled. Examples, include (but of course are not limited to) problems like quantifying the performance of $\ell_2/\ell_1$ (or even $\ell_1$) type of optimization problems in solving systems which on top of having block-sparse solutions also possess other types of structured solution vectors (binary, box-constrained, partially known locations of nonzero blocks, just to name a few), systems with non-exact (noisy) solution vectors and/or equations. In a few forthcoming companion papers we will present some of these applications. However, as it will be clear when these results appear, each of them will require some work to put the mechanism forth but in essence they all will be fairly simple extensions of what we presented in \cite{StojnicUpper10,StojnicCSetamBlock09} and here. The heart of it all will really be the lower-bounding mechanism designed in \cite{StojnicCSetam09,StojnicCSetamBlock09} and the complementary upper-bounding mechanism designed in \cite{StojnicUpper10} and in this paper and how the two ultimately meet in a nice way.

%\newpage1
%\setcounter{page}{1}
\begin{singlespace}
\bibliographystyle{plain}
\bibliography{OptBlDep}
\end{singlespace}

\end{document}